\documentclass[a4paper,11pt]{article}
\usepackage[latin1]{inputenc}
\usepackage[T1]{fontenc}
\usepackage{amsmath, amsthm, amssymb}
\usepackage{array}
\usepackage[final]{graphicx}   
\graphicspath{{tikz/}}

\newcommand{\fq}{{\mathbb{F}_q}}

\newtheorem{theorem}{Theorem}

\newtheorem{lemma}{Lemma}
\newtheorem{proposition}{Proposition}
\newtheorem{corollary}{Corollary}
\theoremstyle{definition}

\newtheorem{defn}{Definition}
\theoremstyle{remark}
\newtheorem{remark}{Remark}
\newtheorem{example}{Example}

\newcommand*{\horzbar}{\rule[.5ex]{2.5ex}{0.5pt}}

\newcommand{\bigslant}[2]{{\raisebox{.1em}{$#1$}\left/\raisebox{-.1em}{$#2$}\right.}}
\usepackage[english]{babel}

\title{Self-Dual Convolutional Codes}
\author{Sebastian Heri,  Julia Lieb and Joachim Rosenthal\\
{\small Institute of Mathematics  \vspace{-1mm}}\\
{\small University of Z\"urich\vspace{-1mm}}\\
{\small Winterthurerstr 190, CH-8057 Z\"urich, Switzerland }
\vspace{3mm} 
}
\date{\today}

\begin{document}
\maketitle

\begin{abstract} 
  This paper investigates the concept of self-dual convolutional code.
  We derive the basic properties of this interesting class of codes
  and we show how some of the techniques to construct self-dual linear
  block codes generalize to self-dual convolutional codes.  As for
  self-dual linear block codes we are able to give a complete
  classification for some small parameters.
\end{abstract}

\section{Introduction}

An $(n,k)$ linear block code $\mathcal{C}$ is by definition a linear
subspace $\mathcal{C}\subset\mathbb{F}_q^n $, where $\mathbb{F}_q$ is
a finite field and $\dim_\fq\mathcal{C}=k.$ By considering the natural
bilinear form $<\, , \, > $ on the vector space $\mathbb{F}_q^n $ one
obtains the notion of the dual code
$$
\mathcal{C}^\perp = \left\lbrace \; x \in \mathbb{F}_q^n \; | \ <x,c>
  = 0 \;\; \forall c \in C \; \right\rbrace.
$$

Clearly $ \mathcal{C}^\perp$ is an $(n-k,k)$ linear block code. Note
also that the bilinear form $<\, , \, > $ is in general not positive
definite and when $n=2k$ it can happen that
$ \mathcal{C}^\perp= \mathcal{C}$. A code having the particular
property that $ \mathcal{C}^\perp= \mathcal{C}$ is called a self-dual
linear block code.

Self-dual block codes are a highly interesting class of linear block
codes and there have been large efforts in the classification of
self-dual codes.  In \cite{27,28} a complete classification of binary
self-dual block codes up to length 30 is provided. This is extended in
\cite{16} to a complete classification up to length 36. 
Moreover, in \cite{15}, \cite{16} and \cite{17}, different techniques how to construct new self-dual block codes from known self-dual block codes are provided. In this paper, we will generalize all of these techniques to obtain new self-dual convolutional codes from known self-dual convolutional codes.
 In general
there exist still many interesting and open questions related to
self-dual block codes and the interested reader is referred to the
survey articles \cite{bou,26}.

In this paper we introduce the concept of a self-dual convolutional
code.  For this note that one can view a convolutional code
$\mathcal{C}$ as a submodule $\mathcal{C}\subset \fq[x]^n $, where
$\fq[x]$ is the polynomial ring over the finite field $\fq$. Again
one has a natural bilinear form $<\, , \, > $ on $\fq[x]^n$ and this
induces the notion of dual convolutional code $\mathcal{C}^\perp $ and
it is therefore again interesting when a convolutional code is
self-dual. In \cite{joh}, examples of particular self-dual convolutional codes are provided, however, using another definition of self-duality than the one presented in our paper. In \cite{14}, an algorithm to find systematic self-dual convolutional codes is presented. However, theoretical approaches concerning properties and classifications of this kind of codes have still been lacking and are the objective of our paper.

The paper is structured as follows. In Section 2, we present the
basics on self-dual block codes and their construction from smaller
self-dual codes. In Section 3, we introduce convolutional codes. In
Section 4, we present criteria for convolutional codes to be self-dual
as well as some properties of self-dual convolutional codes. In
Section 5, we classify all self-dual convolutional codes with certain
code parameters. In Section 5.1, we classify all self-dual $(2,1)$
convolutional codes, in Section 5.2 all binary self-dual $(4,2)$
convolutional codes, in Section 5.3 all self-dual convolutional codes
with double diagonal generator matrices and in Section 5.4 all binary
self-dual convolutional codes with double triangular generator
matrices. In Section 6, we generalize the building-up construction and
the Harada-Munemasa construction from block to convolutional codes,
where the building-up construction is further generalized to arbitrary
finite fields. Moreover, we prove that every binary code obtained with
the generalized building-up construction can be also obtained with the
generalized Harada-Munemasa construction, but the opposite is not
true.  Furthermore, we show that not all binary self-dual
convolutional codes with free distance $d_{free}>2$ can be constructed
with the generalized building-up construction, as it was shown to be
true for block codes. We also show that all binary self-dual $(4,2)$
convolutional codes can be constructed with the generalized
Harada-Munemasa construction, but for general code parameters this
question remains open.

\section{Self-Dual Linear Block Codes}
In this section, we present some results on self-dual linear block
codes, that we will generalize later to obtain results for self-dual
convolutional codes.

\subsection{Preliminaries}

Let $\mathbb{F}_q$ denote the finite field with $q$ elements.

\begin{defn}
  An $(n,k)$ \textbf{(linear) block code} $\mathcal{C}$ over
  $\mathbb{F}_q$ is a $k$-dimensional subspace
  $\mathcal{C} \subset \mathbb{F}_q^n$, where $k \leq n$. A matrix
  $G\in\mathbb F_q^{k\times n}$ such that $\mathcal{C}=rowspan(G)$ is
  called a \textbf{generator matrix} of $\mathcal{C}$.
\end{defn}


\begin{defn}
  We say that two generator matrices $G_1$ and $G_2$ are
  \textbf{equivalent} if there exists an invertible matrix
  $A\in GL_k(\mathbb{F}_q)$ such that
$$G_1=AG_2.$$
In other words, two generator matrices are equivalent if they generate
the same code.
\end{defn}
\begin{defn}
  Let $\mathcal{C}$ be a $(n,k)$ linear block code. Then there exists
  $H\in \mathbb{F}_q^{(n-k) \times n}$ such that
$$\mathcal{C}= ker(H)=\{v\in\mathbb F_q^n\ |\ Hv^{\top}=0\}.$$
We call $H$ a \textbf{parity-check matrix} of $\mathcal{C}$.
\end{defn}
As a generator matrix also a parity-check is in general not unique and
elementary row operations on either of these matrices leaves the code
unchanged.
\begin{defn} For a code $\mathcal{C}\subset\mathbb F_q^n$,
  $\mathcal{C}^\perp = \left\lbrace \; x \in \mathbb{F}_q^n \; | \
    <x,c> = 0 \;\; \forall c \in C \; \right\rbrace $ is called the
  \textbf{dual} of a code $\mathcal{C} \subset \mathbb{F}_q^n$. If the
  condition of $xc^\top=0$ holds, $x$ and $c$ are said to be
  \textbf{orthogonal}.
\end{defn}
The dual of a linear block code is always a linear block code.
\begin{defn} We say that a code $\mathcal{C}$ is
  \textbf{self-orthogonal} if $\mathcal{C} \subset \mathcal{C}^\perp$
  and \textbf{self-dual} if $\mathcal{C}=\mathcal{C}^\perp$.
\end{defn}
Thus, self-orthogonality is fulfilled if all codewords are orthogonal
to each other, whereas self-duality demands the further restriction of
no other element in the vector space $\mathbb{F}_q^n$ being orthogonal
to all codewords.
\\
The following lemma is an immediate consequence of the definition of
the dual code.
\begin{lemma} \cite{23} Let $G$ be a generator matrix of
  $\mathcal{C}$. Then
$$\mathcal{C}^\perp=ker(G)$$
i.e., $G$ is a parity-check matrix of $\mathcal{C}^\perp$.
\end{lemma}
An important consequence of this lemma is that the generator matrix of
a self-dual code is also a parity-check matrix of the given code as
$\mathcal{C}=\mathcal{C}^\perp=ker(G)$.

\begin{corollary} \cite{23} Any self-dual $(n,k)$ code fulfills $n=2k$.
\end{corollary}
Hence, from now on we may assume every self-dual code to be a
$(n,n/2)$ code or equivalently a $(2k,k)$ code.
\begin{lemma}\cite[Chapter 1.8]{mcwsl}
  Let $\mathcal{C}$ be a $(2k,k)$ block code over $\mathbb{F}_q$ with
  generator matrix $G$. Then the following statements are equivalent:
  \begin{itemize}
  \item[(i)] $\mathcal{C}$ is self-dual;
  \item[(ii)] $\mathcal{C}$ is self-orthogonal;
  \item[(iii)] $GG^\top=0$;
  \item[(iv)] $G$ is a parity-check matrix of $\mathcal{C}$.
  \end{itemize}
\end{lemma}
\begin{defn}
  Let $c_1=(x_1,\hdots,x_n)$ and $c_2=(y_1,\hdots,y_n)$ be elements of
  a code. Then their \textbf{(Hamming) distance} is defined to be
$$d(c_1,c_2):=\left\lbrace i|x_i \neq y_i\right\rbrace $$
and the \textbf{(Hamming) weight} of $c_1$ is $wt(c_1):=d(c_1,0)$.
\end{defn}
\begin{defn}
  The \textbf{(minimum) distance} d of a linear block code
  $\mathcal{C}$ is given by
$$d=d(C)=min\left\lbrace wt(c) \ | \ c \in \mathcal{C} \setminus \{0\} \right\rbrace .$$
One uses the notation $[n,k,d]$ code for an $(n,k)$ block code with
distance $d$.
\end{defn}

%

\subsection{Construction Methods}
There are three popular construction methods for binary self-dual
codes, namely the building-up construction, the Harada-Munemasa
construction and the recursive construction. The latter will not be
presented, but interested readers may consult \cite{19}.
\begin{theorem}[\textbf{Building-up construction}]\label{bu} \cite{17}
  Let $\mathcal{C}$ be a binary self-dual $(2k,k)$ code and
  $G=(g_i)_{i\in \left\lbrace 1,...,k\right\rbrace }$ its generator
  matrix, where $g_i$ is the $i$-th row of $G$. Let
  $x \in \mathbb{F}_2^n$ be a binary vector with odd weight and define
  $y_i:=xg_i^\top$ for $1\leq i \leq k$. Then
$$\tilde{G}=
\begin{pmatrix}
  1 & 0 & x \\
  y_1 & y_1 & \\
  \vdots & \vdots & G \\
  y_k & y_k &
\end{pmatrix}
\in \mathbb{F}_2^{(k+1) \times (2k+2)}
$$
generates a binary self-dual $(2k+2,k+1)$ code $\tilde{\mathcal{C}}$.
\end{theorem}

%
\begin{theorem}\cite{17}\label{abu}
  Any binary self-dual code of length $n$ with distance $d>2$ can be
  obtained from some binary self-dual code of length $n-2$ with the
  construction in Theorem \ref{bu} (up to column permutations in the
  generator matrix).
\end{theorem}

As second construction method we present the Harada-Munemasa
construction.
\begin{theorem}[\textbf{Harada-Munemasa construction}]\label{hm}
  \cite{16} Let $G$ be the generator matrix of a binary self-dual
  $[2k,k,d]$ code $\mathcal{C}$. Then the matrix
$$
G_1=
\begin{pmatrix}
  a_1 & a_1 & \\
  \vdots & \vdots & G \\
  a_k & a_k &
\end{pmatrix} \in \mathbb{F}_2^{k \times (2k+2)},
$$
where $a_i \in \mathbb{F}_2$, generates a binary self-orthogonal
$[2k+2,k,\geq d]$ code $\mathcal{C}_1$. Moreover, there exists
$x \in \mathcal{C}_1^\perp \setminus \mathcal{C}_1$ such that
$$
G_2=
\begin{pmatrix}
  \horzbar & x & \horzbar \\
  a_1 & a_1 & \\
  \vdots & \vdots & G \\
  a_k & a_k &
\end{pmatrix} \in \mathbb{F}_2^{(k+1) \times (2k+2)}
$$
generates a binary self-dual $(2k+2,k+1)$ code.
\end{theorem}
\begin{theorem}
  Any binary self-dual code of length $n$ with distance $d>2$ can be
  obtained by some binary self-dual code of length $n-2$ with the
  construction in Theorem \ref{hm} (up to column permutations in the
  generator matrix).
\end{theorem}

The difference between the two constructions is that the building-up
construction chooses the row first and then simply calculates the
entries of the new columns based on the chosen row, whereas the
Harada-Munemasa construction chooses the columns first and then
determines the to be added row.

\section{Convolutional Codes}
In this section, we first introduce some basics about convolutional
codes with a particular focus on non-catastrophic convolutional codes.
Then, we will connect non-catastrophic codes to self-dual codes and
present equivalent properties to self-duality, which will be used for
classifications and constructions in later chapters.
\subsection{Preliminaries}
\begin{defn}
  An \textbf{$(n,k)$ convolutional code} is a
  $\mathbb{F}_q[z]$-submodule of $\mathbb{F}_q[z]^n$ with rank $k$.
\end{defn}
$\mathbb{F}_q[z]$ is a PID and modules over a PID always admit a
basis. Therefore, one always finds
$G(z)\in \mathbb{F}_q[z]^{k \times n}$ such that
$\mathcal{C}=rowspan(G(z))$. This $G(z)$ is called a \textbf{generator
  matrix} of $\mathcal{C}$.
\begin{defn}
  Let $H(z)\in \mathbb{F}_q[z]^{(n-k) \times n}$ such that
$$\mathcal{C}=ker(H(z)).$$
Then we call $H(z)$ a \textbf{parity-check matrix} of $\mathcal{C}$.
\end{defn}
Other than for block codes, there are some convolutional codes that do
not admit a parity-check matrix. More details on the specific
characteristics that permit the existence of a parity-check matrix
will be presented in Chapter 3.2.
\begin{defn}
  We say that $U(z)\in\mathbb{F}_q[z]^{k \times k}$ is
  \textbf{unimodular} if there exists
  $V(z)\in\mathbb{F}_q[z]^{k \times k}$ such that
$$V(z)U(z)=I_k.$$
\end{defn}
\begin{defn}
  Two generator matrices $G(z)\in \mathbb{F}_q[z]^{k \times n}$ and
  $\tilde{G}(z) \in \mathbb{F}_q[z]^{k \times n}$ are
  \textbf{equivalent} if there exists a unimodular matrix
  $U(z) \in \mathbb{F}_q[z]^{k \times k}$ such
  that $$G(z)=U(z)\tilde{G}(z).$$
\end{defn}
Equivalent generator matrices generate the same code.

\begin{defn}
  We say that a row or column operation is \textbf{unimodular} if it
  is done via multiplication with a unimodular matrix.
\end{defn}

Furthermore, we introduce a notion that in a way tells us how ``far
away" a convolutional code is from a block code.
\begin{defn}
  The \textbf{degree $\delta$} of a convolutional code is defined to
  be the highest degree of any $k$-th minor of its generator matrix.
\end{defn}
As the generator matrices of two equivalent codes differ by
multiplication with a unimodular matrix, which has a constant
determinant, the degree $\delta$ is the same for all equivalent
generator matrices, making it a well-defined notion.
\begin{remark}
  The convolutional codes of degree $0$ are essentially the linear
  block codes.
\end{remark}
\begin{defn}
  The \textbf{weight} of a polynomial vector
$$c(z)=\sum_{i=0}^{deg(c(z))}c_iz^i \in \mathbb{F}_q[z]^n$$
is defined to be
$$wt(c(z)):=\sum_{i=0}^{deg(c(z))}wt(c_i),$$
where $wt(c_i)$ denotes the Hamming weight of $c_i$.
\end{defn}
\begin{defn} The \textbf{free distance} of a convolutional code
  $\mathcal{C}$ is defined as
$$d_{free}(\mathcal{C})=min\left\lbrace wt(c(z)) \ | \ c(z) \in \mathcal{C}\setminus \{0\} \right\rbrace.$$
\end{defn}
\subsection{Non-Catastrophic Codes}
\begin{defn} A polynomial matrix
  $G(z) \in \mathbb{F}_q[z]^{k \times n}$ is said to be
  \textbf{left-prime} if in all factorizations
$$G(z)=A(z)\tilde{G}(z), \textit{ with } A(z)\in \mathbb{F}_q[z]^{k \times k} \textit{ and } \tilde{G}(z)\in \mathbb{F}_q[z]^{k \times n},$$
the left factor $A(z)$ is unimodular.
\end{defn}
If one generator matrix is left-prime, then all equivalent generator
matrices are also left-prime as they differ by left-multiplication
with a unimodular matrix. We call codes that admit a left-prime
generator matrix \textbf{non-catastrophic} and \textbf{catastrophic}
if they do not.

We will see more practical methods to check whether codes are
non-catastrophic in Theorem \ref{catcrit}, but first we introduce some
necessary notions.
\begin{defn}\cite[Section 2.1]{gantmacher1977}
  Let $G(z) \in \mathbb{F}_q[z]^{k \times n}$ with $ k \leq n $. Then,
  there exists a unimodular matrix
  $U(z) \in \mathbb{F}_q[z]^{k \times k}$ such that
$$ G_{rH}(z)=U(z)G(z)= 
\begin{pmatrix}
  h_{11}(z) & h_{12}(z) & \cdots & h_{1k}(z) & h_{1{k+1}}(z) & \cdots & h_{1{2k}}(z) \\
  & h_{22}(z) & \cdots & h_{2k}(z) & h_{2{k+1}}(z) & \cdots & h_{2{2k}}(z) \\
  & & \ddots & \vdots & \vdots & & \vdots\\
  & & & h_{kk}(z) & h_{{k}{k+1}}(z) & \cdots & h_{k{2k}}(z)
\end{pmatrix},
$$
where $(h_{ii}(z))_{i=1}^k$ are monic polynomials such that
$deg(h_{ii}(z)) > deg(h_{ji}(z))$ for $i>j$ if $h_{ii}(z)$ is not
identically equal to zero. The matrix $G_{rH}(z)$ is unique and called
\textbf{row Hermite form} of $G(z)$.
\end{defn}

%
%

The uniqueness of the row Hermite form gives us the option to
characterize equivalent generator matrices by their respective row
Hermite form.
\begin{defn}
  Let $G(z) \in\mathbb{F}_q[z]^{k \times n}$ with $ k \leq n$. Then
  there exists a unimodular matrix
  $V(z) \in \mathbb{F}_q[z]^{n \times n}$ such that
$$ G_{cH}(z)=G(z)V(z)= 
\begin{pmatrix}
  h_{11}(z) & & & & 0 & \cdots & 0 \\
  h_{21}(z) & h_{22}(z) & & & 0 & \cdots & 0 \\
  \vdots & \vdots & \ddots & & \vdots & & \vdots \\
  h_{k1}(z) & h_{k2}(z) & \cdots & h_{kk}(z) & 0 & \cdots & 0
\end{pmatrix},
$$
where $(h_{ii}(z))_{i=1}^k$ are monic polynomials such that
$deg(h_{ii}(z)) > deg(h_{ij}(z))$ for $i>j$ if $h_{ii}(z)$ is not
identically equal to zero. $G_{cH}(z)$ is called the (unique)
\textbf{column Hermite Form} of $G(z)$.
\end{defn}

\begin{defn}
  Let $G(z) \in \mathbb{F}_q[z]^{k \times n}$ be full rank with
  $ k \leq n$. Then there exist two unimodular matrices
  $U(z) \in \mathbb{F}_q[z]^{k \times k}$ and
  $V(z) \in \mathbb{F}_q[z]^{n \times n}$ such that
$$ S(z)=U(z)G(z)V(z)= 
\begin{pmatrix}
  \gamma_1(z) & & & & 0 & \cdots & 0 \\
  & \gamma_{2}(z) & & & 0 & \cdots & 0 \\
  & & \ddots & & \vdots & & \vdots \\
  & & & \gamma_{k}(z) & 0 & \cdots & 0
\end{pmatrix},
$$
where $(\gamma_{i}(z))_{i=1}^k$ are monic polynomials such that
$\gamma_{i+1}(z) | \gamma_{i}(z)$ for all $i=1,...,k-1$. $S(z)$ is
called the (unique) \textbf{Smith form} of $G(z)$.
\end{defn}
The following result will be crucial in later sections to determine
whether a code is non-catastrophic.
\begin{theorem}\label{catcrit}
  \cite{1,2} Let $\mathcal{C}$ be an $(n,k)$convolutional code and
  $G(z) \in \mathbb{F}_q[z]^{k \times n}$ its generator matrix. The
  following are equivalent:
  \begin{itemize}
  \item[1.] $\mathcal{C}$ is non-catastrophic;
  \item[2.] $G(z)$ is left-prime;
  \item[3.] the Smith form of $G(z)$ is $[I_k \ 0]$;
  \item[4.] the column Hermite form of $G(z)$ is $[I_k \ 0]$;
  \item[5.] the ideal generated by all the $k$-th minors of $G(z)$ is
    $\mathbb{F}_q[z]$;
  \item[6.] there exists a parity-check matrix $H(z)$ of
    $\mathcal{C}$;
  \item[7.] $G(z)$ can be completed to an unimodular matrix,
    e.g. there exists $L(z) \in \mathbb{F}_q[z]^{(n-k) \times n}$ such
    that $[\begin{smallmatrix}
      G(z) \\
      L(z)
    \end{smallmatrix}]$ is unimodular.
  \end{itemize}
\end{theorem}
\begin{corollary}
  Let $G(z)$ be the generator matrix of a non-catastrophic $(n,k)$
  convolutional code. Then, for any unimodular
  $U(z) \in\mathbb F_q[z]^{k \times k}$ and
  $V(z) \in \mathbb F_q[z]^{n \times n}$,
$$U(z)G(z)V(z)$$
also generates a non-catastrophic code. Or in other words, a
non-catastrophic code stays non-catastrophic under unimodular
operations.
\end{corollary}
This corollary is due to the fact that the Smith normal form does not
change when multiplying it with unimodular matrices. It further shows
that we may use unimodular operations on a generator matrix to make it
easier to determine either of the properties presented in Theorem
\ref{catcrit}.
\begin{defn}
  For two matrices $G_1(z),G_2(z) \in \mathbb F_q[z]^{k\times n}$, we
  write $$G_1(z) \sim G_2(z),$$ if there exist unimodular matrices
  $A(z)\in \mathbb F_q[z]^{k \times k}$ and
  $B(z) \in \mathbb F_q[z]^{n \times n}$ such that
$$G_1(z)=A(z)G_2(z)B(z).$$
\end{defn}

\section{Self-Dual Convolutional Codes}
In this section, we present the basic theory on self-dual
convolutional codes.

\begin{defn}Let $\mathcal{C}$ be an $(n,k)$ convolutional
  codes. Then,
  $$\mathcal{C}^\perp = \left\lbrace f(z) \in R^n \ | \ f(z)c(z)^\top
    = 0 \ \ \ \forall c(z) \in C \right\rbrace $$ is called the
  \textbf{dual} of $\mathcal{C}$.
\end{defn}
\begin{lemma}\label{dc} \cite{14}
  The dual of a convolutional code $\mathcal{C}$ is also a
  convolutional code.
\end{lemma}
\begin{defn}
  Two vectors $u(z),v(z) \in \mathbb F_q[z]^n$ are said to be
  \textbf{orthogonal} if $$u(z)v(z)^\top=0.$$
\end{defn}

\subsection{Characterization of self-dual convolutional codes over
  arbitrary finite fields}

The following lemma establishes that the dual code always has a
parity-check matrix.
\begin{lemma}\label{dp}
  Let
$$G(z)=
\begin{pmatrix}
  \horzbar & g_1(z) & \horzbar \\
  & \vdots &  \\
  \horzbar & g_k(z) & \horzbar
\end{pmatrix}$$ be the generator matrix of a convolutional code
$\mathcal{C}$, then
$$\mathcal{C}^\perp=ker(G(z))$$
or equivalently $G(z)$ is a parity-check matrix of
$\mathcal{C}^\perp$.
\end{lemma}
\begin{proof}
  To show $\mathcal{C}^\perp \subset ker(G(z))$, let
  $c(z) \in \mathcal{C}^\perp$. One has $g_i(z)c(z)^\top=0$ for
  $i=1,\hdots,k$ and therefore
  $G(z)c(z)^\top=0$, i.e., $c(z)\in ker(G(z))$.\\
  To show $ker(G(z)) \subset \mathcal{C}^\perp$, let
  $f(z)\in ker(G(z))$ and $c(z) \in \mathcal{C}$. Then,
  $c(z)=m(z)G(z)$ for some $m(z) \in \mathbb{F}_q[z]^k$. But now
$$c(z)f(z)^\top=(m(z)G(z))f(z)^\top=m(z)(G(z)f(z)^\top)=0$$
as $f(z)$ is in the kernel of $G(z)$. We conclude that $f(z)$ is
orthogonal to every codeword of $\mathcal{C}$ or equivalently
$f(z) \in \mathcal{C}^\perp$.
\end{proof}
Thus, any generator matrix of a convolutional code is a parity-check
matrix of the dual code. Consequently, the dual is always
non-catastrophic by Theorem \ref{catcrit}.
\begin{lemma}\label{dd}
  Let $\mathcal{C}$ be an $(n,k)$ convolutional code with generator
  matrix $G(z)$. Then,
$$\mathcal{C} \text{ is non-catastrophic}$$
if and only if
$$\mathcal{C}=(\mathcal{C}^\perp)^\perp.$$
\end{lemma}
\begin{proof}
  $(\Longrightarrow)$ Let $c(z) \in \mathcal{C}$ and
  $c_1(z) \in \mathcal{C}^\perp$, then by definition
  $c_1(z)c(z)^\top=0$ and hence
$$\mathcal{C} \subset (\mathcal{C}^\perp)^\perp.$$
As $(\mathcal{C}^\perp)^\perp$ is again a convolutional code, it has a
generator matrix $\tilde{G}(z)$ with
$rowspan(G(z)) \subset rowspan(\tilde{G}(z))$. Moreover,
$dim((\mathcal{C}^\perp)^\perp)=n-dim(\mathcal{C}^\perp)=dim(\mathcal{C})=k$
and one obtains $U(z)\tilde{G}(z)=G(z)$ for some
$U(z)\in\mathbb F_q[z]^{k\times k}$.
%
%
But $G(z)$ is left-prime, because we assumed that
$\mathcal{C}$ is non-catastrophic, and hence the left-factor
$U(z)$ is unimodular.  We conclude that $G(z)$ and
$\tilde{G}(z)$ are equivalent and
$$\mathcal{C}=(\mathcal{C}^\perp)^\perp.$$
$(\Longleftarrow)$
Follows directly from the previous lemma as
$\mathcal{C}=(\mathcal{C}^\perp)^\perp$ is the dual of
$\mathcal{C}^\perp$.
\end{proof}

\begin{defn} We say that a convolutional code
  $\mathcal{C}$ is \textbf{self-orthogonal} if $\mathcal{C} \subset
  \mathcal{C}^\perp$ and \textbf{self-dual} if
  $\mathcal{C}=\mathcal{C}^\perp$.
\end{defn}
The following lemma follows immediately.
\begin{lemma}\label{noncat} Let $\mathcal{C}$ be a self-dual
  $(n,k)$ convolutional code with generator matrix
  $G(z)$. Then, $\mathcal{C}$ is non-catastrophic.
\end{lemma}

Like for block codes, we may assume every self-dual convolutional code
to be a $(2k,k)$ or equivalently an $(n,n/2)$ code as stated in the
following lemma.

\begin{lemma} \cite{14} If $\mathcal{C}$ is a self-dual $(n,k)$
  convolutional code, then $n=2k$.
\end{lemma}

In a next step, we will establish the main result of this chapter,
which will be used for determining whether a convolutional code is
self-dual in later sections.

\begin{theorem}\label{sdcrit}
  Let $\mathcal{C}$ be a $(2k,k)$ convolutional code with generator
  matrix $G(z)$, then the following statements are equivalent:
  \begin{itemize}
  \item[(i)] $\mathcal{C}$ is self-dual;
  \item[(ii)] $\mathcal{C}$ is non-catastrophic and self-orthogonal;
  \item[(iii)] $\mathcal{C}$ is non-catastrophic and
    $G(z)G(z)^\top=0$;
  \item[(iv)] $G(z)$ is a parity-check matrix of $\mathcal{C}$.
  \end{itemize}
\end{theorem}
\begin{proof}
  $(i) \Rightarrow (ii)$ Let $\mathcal{C}$ be self-dual, then the code
  is self-orthogonal by definition and non-catastrophic by Lemma
  \ref{noncat}.
  \\
  $(ii) \Rightarrow (i)$ Let $\mathcal{C}$ be non-catastrophic and
  self-orthogonal and let $H(z)$ be a generator matrix of
  $\mathcal{C}^\perp$.  $\mathcal{C}\subset \mathcal{C}^\perp$ implies
  $rowspan(G(z)) \subset rowspan(H(z))$ and hence $U(z)H(z)=G(z)$ for
  some $U(z)\in\mathbb F_q[z]^{k\times k}$.  We assumed $\mathcal{C}$
  to be non-catastrophic, i.e., $G(z)$ is left-prime and hence, $U(z)$
  must be unimodular.  But this means that $G(z)$ and $H(z)$ are
  generator matrices of the same code and
  $\mathcal{C}=\mathcal{C}^\perp.$
  \\
  $(ii) \Rightarrow (iii)$ Let $\mathcal{C}$ be self-orthogonal with
  generator matrix $G(z)$. Each pair of rows $g_i(z),g_j(z)$ of $G(z)$
  is orthogonal, i.e.,$g_i(z)g_j(z)^\top=0$ and $G(z)G(z)^\top=0$
  follows.
  \\
  $(iii) \Rightarrow (ii)$ Let $\mathcal{C}$ fulfil
  $G(z)G(z)^\top=0$. Take $c_1(z),c_2(z) \in \mathcal{C}$, then there
  exist $m_1(z),m_2(z)\in \mathbb{F}_q[z]^k$ such that
$$\begin{array}{cc}
    c_1(z)&= m_1(z)G(z) \\
    c_2(z)&= m_2(z)G(z)
  \end{array}.$$
  But now 
  $c_1(z)c_2(z)^\top=m_1(z)(G(z)G(z)^\top)m_2(z)^\top=0$
  and therefore each pair of codewords is orthogonal or equivalently $\mathcal{C}$ is self-orthogonal. \\
  $(i) \Longrightarrow (iv)$ By assumption,
  $\mathcal{C}=\mathcal{C}^\perp$ and by Lemma \ref{dp},
$$\mathcal{C}^\perp=ker(G(z)),$$
making $G(z)$ a parity check matrix of $\mathcal{C}$.
\\
\\
$(iv) \Longrightarrow (ii)$ Let $G(z)$ be a parity-check of
$\mathcal{C}$, then $\mathcal{C}$ is non-catastrophic
and $G(z)G(z)^\top=0$, because $rowspan(G(z))=\mathcal{C}=ker(G(z))$.
\end{proof}
One might observe that self-orthogonality and $G(z)G(z)^\top=0$ are equivalent statements as the assumption of the code being non-catastrophic was not used in the proof of $(ii) \Longleftrightarrow (iii)$. \\
\\
For $(2k,k)$ linear block codes, self-duality and self-orthogonality
are equivalent properties, which is due to the fact that each linear
block code possesses a parity-check matrix. For $(2k,k)$ convolutional
codes however, this is not necessarily true as illustrated in the
following example.
\begin{example}
  Let
$$G(z)=
\begin{pmatrix}
  z^2+z+1 & z^2 & z & 1 \\
  1 & z & z^2 & z^2+z+1
\end{pmatrix} \in \mathbb{F}_2[z]^{2 \times 4}$$ be the generator
matrix of a binary $(4,2)$ convolutional code $\mathcal{C}$.
\\
We want to show that $\mathcal{C}=rowspan(G(z))$ is a proper subset of $\mathcal{C}^\perp=ker(G(z))$. \\
As $G(z)G(z)^\top=0$,
$\mathcal{C} \subset \mathcal{C}^\perp.$ Moreover, we observe that
$(1,1,1,1) \in \mathcal{C}^\perp$
but $(1,1,1,1)\not \in \mathcal{C}$, which shows that $\mathcal{C}$ is indeed a proper subset of $\mathcal{C}^\perp$. \\
Conclusively, $\mathcal{C}$ is a self-orthogonal $(2k,k)$
convolutional codes that is not self-dual. As a consequence of Theorem
\ref{sdcrit}, $\mathcal{C}$ must be a catastrophic code.
\end{example}


\subsection{Special properties of binary self-dual convolutional
  codes}

In this subsection, we present some further results for binary
self-dual convolutional codes.

The following well-known lemma is an immediate consequence of the
"Freshman dream" and helps us to determine whether a binary polynomial
vector is orthogonal to itself.

\begin{lemma}\label{fd}
  Let $f_i(z)\in \mathbb{F}_2[z]$, for
  $i\in \left\lbrace 1,\hdots,m\right\rbrace $, and
  $f(z)=(f_1(z),\hdots,f_m(z))$, then
$$f(z)f(z)^\top=(f(z)(1,\hdots,1)^\top)^2.$$
\end{lemma}
As a consequence, to determine whether a binary polynomial vector is
orthogonal to itself, we may just check if the sum of its entries is
equal to zero.
\begin{lemma}\label{ew}
  All codewords of a binary self-dual code have even weight.
\end{lemma}
\begin{proof}
  Let $c(z)=\sum_{i=0}^l c_iz^i$ be a codeword of a binary self-dual
  $(n,k)$ convolutional code, where $l={deg(c(z))}$ and
  $c_i\in \mathbb{F}_2^n$. By definition of self-duality,
  one obtains $c(z)c(z)^\top=0$.  It follows from Lemma \ref{fd} that
  $c(z)(1,\hdots,1)^\top=0$ or equivalently
$$\sum_{i=0}^l c_iz^i (1,\hdots,1)^\top= \sum_{i=0}^l z^ic_i(1,\hdots,1)^\top=0.$$
Every coefficient of the last sum must be equal to zero, i.e.
$$
c_i(1,\hdots,1)^\top=0
$$
must hold for every $i\in \{0,\hdots,l\}$ and hence $wt(c_i)$ is
even. But then
$$wt(c(z))=\sum_{i=0}^l wt(c_i)$$
is also even.
\end{proof}
In \cite{22}, the following result was proven for block codes. Up to
our knowledge it has not been proven for convolutional codes yet.
\begin{lemma}\label{11} Every binary self-dual convolutional code
  contains $(1,\hdots,1) \in \mathbb{F}_2^n$.
\end{lemma}
\begin{proof} Let $g_i(z)$ be a row of the generator matrix $G(z)$,
  then $g_i(z)g_i(z)^\top=0$ or equivalently (Lemma \ref{fd}),
  $g_i(z)(1,\hdots,1)^\top=0$ and therefore $G(z)(1,\hdots,1)^\top=0$.
  Hence $(1,\hdots,1)$ is contained in the dual of the code and,
  because of self-duality, contained in the code itself.
\end{proof}


\begin{corollary}
  Let $\mathcal{C}$ be a binary self-dual $(n,k)$ convolutional code,
  then there exists a generator matrix whose first row is the all one
  vector.
\end{corollary}
This generator matrix with just ones in the first row is equivalent to
its row Hermite form. Therefore, to find all binary self-dual
convolutional codes, we may characterize them via generator matrices
that have a first row of just ones and are in row Hermite form.
\begin{remark}\label{repform}
  Let
$$G(z)=
\begin{pmatrix}
  1 & 1 & 1 & 1 & \cdots & \cdots & 1 \\
  0 & g_{22}(z) & g_{23}(z) & g_{24}(z) & \cdots & \cdots & g_{2n}(z) \\
  \vdots & 0 & g_{33}(z) & g_{34}(z) & \cdots & \cdots & g_{3n}(z) \\
  \vdots & \vdots & \ddots & \ddots & \ddots & & \vdots \\
  0 & 0 & \cdots & 0 & g_{kk}(z) & \cdots & g_{kn}(z)
\end{pmatrix}$$ be a generator matrix of a self-dual code in the just
discussed form. Then we choose this form of $G(z)$ as the
representative of any binary self-dual convolutional code for our
classifications.
\end{remark}
Multiplying the generator matrix of a self-dual code with a unimodular
matrix from the left-hand side leaves the code unchanged and therefore
leaves the code self-dual.  Obviously, exchanging columns has also no
influence on self-duality since it only results in renumbering the
components of the codewords.  However, the following example shows
that multiplying a column with a constant or adding columns does not
preserve self-duality.
\begin{example}
  We claim that
$$G(z)=\begin{pmatrix}
  3 & z & 1 & 3z \\
  1 & 2z+4 & 2 & z+2
\end{pmatrix} \in \mathbb{F}_5[z]^{2 \times 4}
$$
generates a self-dual code and will use Theorem \ref{sdcrit} to proof
this. Firstly,
$$
G(z)G(z)^\top=
\begin{pmatrix}
  10z^2+10 & 5z^2+10z+5 \\
  5z^2+10z+5 & 5z^2+20z+25
\end{pmatrix}
=0
$$
and secondly
$$
det\left[
  \begin{pmatrix}
    3 & z \\
    1 & 2z+4
  \end{pmatrix}\right]
= 6z+12-z=2 ,
$$
which implies that the set of $G(z)$`s full-size minors generates $\mathbb{F}_5[z]$ or equivalently $G(z)$ generates a non-catastrophic code. Hence, the code is self-dual. \\
We now want to establish that adding and multiplying columns might
change this property. For this, add the second column of $G(z)$ to the
first to find
$$\begin{pmatrix}
  3 & z & 1 & 3z \\
  1 & 2z+4 & 2 & z+2
\end{pmatrix}\sim
\begin{pmatrix}
  3+z & z & 1 & 3z \\
  2z & 2z+4 & 2 & z+2
\end{pmatrix}=:G_1(z).
$$
Then the first row of $G_1(z)$ is not orthogonal to itself as
$$(3+z,z,1,3z)(3+z,z,1,3z)^\top=9+6z+z^2+z^2+1+9z^2=z^2+z \neq 0 ,$$
meaning that $G_1(z)$ does not generate a self-dual code. \\
Similarly, if we multiply the first column of $G(z)$ by $2$ we find
$$
\begin{pmatrix}
  3 & z & 1 & 3z \\
  1 & 2z+4 & 2 & z+2
\end{pmatrix}
\sim
\begin{pmatrix}
  1 & z & 1 & 3z \\
  2 & 2z+4 & 2 & z+2
\end{pmatrix}
$$
and again the first row is not orthogonal to itself as
$$(1,z,1,3z)(1,z,1,3z)^\top=1+z^2+1+9z^2=2\neq 0 .$$
\end{example}
\section{Classification of Self-Dual Convolutional Codes}
In this section, we will classify all self-dual $(2,1)$ convolutional
codes over finite fields and all binary self-dual $(4,2)$
convolutional codes. We will further classify all self-dual
convolutional codes with double diagonal matrices over finite fields
and all binary self-dual convolutional codes with double upper
triangular generator matrices.
\subsection{Self-Dual (2,1) Convolutional Codes}
In this subsection, we consider $\mathbb{F}_q[z]$, where $q=p^l$ is a
prime power.

\begin{theorem}\label{1/2}
  The self-dual $(2,1)$ convolutional codes are exactly the $(2,1)$
  self-dual block codes. In particular, these are exactly the codes
  with a generator matrix of the form $(a\ b)$ where $a^2=-b^2$ in
  $\mathbb F_q$. Clearly, such a generator matrix exists if and only
  if $-1$ is a square in $\mathbb F_q$, i.e., for $p\equiv 1\mod 4$,
  $p=2$ or $l$ even.
\end{theorem}

\begin{proof}
  Denote by $G(z)=(g_1(z),g_2(z)) \in \mathbb F_q[z]^{1 \times 2}$ a
  generator matrix of a self-dual $(2,1)$ convolutional code
  $\mathcal{C}$. By Theorem \ref{sdcrit}, the code is self-dual if and
  only if $g_1(z)^2+g_2(z)^2=0$ and it is non-catastrophic. However,
  such a code is non-catastrophic if and only if $g_1(z)$ and $g_2(z)$
  are coprime, which can only be fulfilled together with
  $g_1(z)^2=-g_2(z)^2$ if $g_1(z)$ and $g_2(z)$ are constant.
\end{proof}
\subsection{Binary Self-Dual (4,2) Convolutional Codes}
The goal in this section is to find all binary self-dual $(4,2)$
convolutional codes $\mathcal{C}$ classified by their representative
as defined in Remark 3.3.16., i.e.,
%
we may assume that
$$G(z)=
\begin{pmatrix}
  1 & 1 & 1 & 1 \\
  0 & g_{22}(z) & g_{23}(z) & g_{24}(z)
\end{pmatrix} \in \mathbb{F}_2[z]^{2 \times 4}$$ is a generator matrix
of $\mathcal{C}$ for some
$g_{22}(z),g_{23}(z),g_{24}(z) \in \mathbb{F}_2[z]$.
$G(z)G(z)^\top=0$ tells us that
$g_{22}(z)^2+g_{23}(z)^2+g_{24}(z)^2=0$ or equivalently (by Lemma
\ref{fd}),
\begin{equation}\label{1}
  g_{22}(z)+g_{23}(z)+g_{24}(z)=0.
\end{equation}
Moreover, $\mathcal{C}$ is non-catastrophic if and only if the
elements of the set of $2$-nd minors
$M=\left\lbrace
  g_{22}(z),g_{23}(z),g_{24}(z),g_{22}(z)+g_{23}(z),g_{22}(z)+g_{24}(z),g_{23}(z)+g_{24}(z)\right\rbrace$
of $G(z)$ are coprime, i.e.,
$gcd(g_{22}(z),g_{23}(z),g_{24}(z))=1$.  Using $\eqref{1}$ this
simplifies to $gcd(g_{23}(z),g_{24}(z))=1.$
\\
We conclude that for $g_{23}(z),g_{24}(z)\in \mathbb{F}_2[z]$ such
that $gcd(g_{23}(z),g_{24}(z))=1$,
$$G(z)=
\begin{pmatrix}
  1 & 1 & 1 & 1 \\
  0 & g_{23}(z)+g_{24}(z) & g_{23}(z) & g_{24}(z)
\end{pmatrix}$$ generates a binary self-dual $(4,2)$ convolutional
code and that every binary self-dual $(4,2)$ convolutional code has a
generator matrix of this form.
\subsection{Double Diagonal Generator Matrix}
In this subsection, we want to find all self-dual convolutional codes
with generator matrices of the form
$$
G(z) =
\begin{pmatrix}
  g_{1,1}(z) & & & & g_{1,{k+1}}(z) & & &  \\
  & g_{2,2}(z) & & & & g_{2,{k+1}}(z)  & & \\
  & & \ddots & & & & \ddots & \\
  & & & g_{k,k}(z) & & & & g_{k,{2k}}(z)
\end{pmatrix} \in \mathbb{F}_q^{k \times 2k}.
$$
\\
Let $i \in \left\lbrace 1,...,k\right\rbrace $. Self-orthogonality is
equivalent to
$g_{i,i}(z)^2+g_{i,{k+i}}(z)^2=0$
and non-catastrophicity is equivalent to
$\gcd(g_{i,i}(z),g_{i,k+i}(z))=1$.
\\
As in Theorem \ref{1/2}, one concludes that
$$
G(z) =
\begin{pmatrix}
  a_1 & & & & b_1 & & &  \\
  & a_2 & & & & b_2  & & \\
  & & \ddots & & & & \ddots & \\
  & & & a_k & & & & b_k
\end{pmatrix} \in \mathbb{F}_q^{k \times 2k},
$$
where $a_i,b_i \in \mathbb{F}_{q}^*$ such that $a_i^2=-b_i^2$ in
$\mathbb F_q$. Therefore, self-dual convolutional codes of this form
exist for exactly the same field sizes as $(2,1)$ self-dual block
codes.
\subsection{Double Upper Triangular Generator Matrices}
Let $\mathcal{C}$ be a binary self-dual $(2k,k)$ convolutional code
with a generator matrix $G(z)$ of the form
$$
G(z) =
\begin{pmatrix}
  g_{1,1} & \cdots & \cdots & g_{1,k} & g_{1,{k+1}} & g_{1,{k+2}} & \cdots & g_{1,{2k}} \\
  & g_{2,2} & \cdots & g_{2,k} & & g_{2,{k+2}}  & \cdots & g_{2,{2k}} \\
  & & \ddots & \vdots & & & \ddots & \vdots \\
  & & & g_{k,k} & & & & g_{k,{2k}}
\end{pmatrix} \in \mathbb{F}_2[z]^{k \times 2k}.
$$
For simplicity we did and will write $g_{i,j}$ instead of
$g_{i,j}(z)$.  We want to show that
$$
\begin{pmatrix}
  1 & & & & 1 & & &  \\
  & 1 & & & & 1  & & \\
  & & \ddots & & & & \ddots & \\
  & & & 1 & & & & 1
\end{pmatrix}
$$
is also a generator matrix of $\mathcal{C}$.  Let $g_i$ be the $i$-th
row of $G(z)$, then $g_k g_i^\top=0$ for
$i \in \left\lbrace 1,\hdots , k\right\rbrace $ leads to
$$
\begin{array}{ccc}
  g_{k,k} & = & g_{k,{2k}}, \\
  g_{k,k}g_{i,k}+ g_{k,{2k}}g_{i,2k} & = &g_{k,k}(g_{i,k}+g_{i,2k})  =  0.\\
\end{array} 
$$
$g_{k,k}$ is non-zero (as $rank(G(z))=k$) and hence
$g_{i,k}=g_{i,2k}$.
Moreover, $g_{k,k}$ is a common factor of every $k$-th minor of $G(z)$
and hence must be equal to $1$ because the code is
non-catastrophic. Therefore,
$$
G(z) =
\begin{pmatrix}
  g_{1,1} & g_{1,2} & \cdots & \cdots & g_{1,k} & g_{1,k+1} & g_{1,{k+2}} & \cdots & \cdots & g_{1,k} \\
  & g_{2,2} & g_{2,3} & \cdots & g_{2,k} & & g_{2,k+2}  & g_{2,k+3} & \cdots & g_{2,k} \\
  & & \ddots & \ddots & \vdots & & & \ddots & \vdots \\
  & & & g_{k-1,k-1} & g_{k-1,k} & & & & g_{k-1,2k-1} & g_{k-1,k} \\
  & & & & 1 & & & & & 1
\end{pmatrix}.
$$
Furthermore, $G(z)$ is equivalent to
$$
\tilde{G}(z)=\begin{pmatrix}
  g_{1,1} & g_{1,2} & \cdots & g_{1,k-1} & 0 & g_{1,k+1} & g_{1,{k+2}} & \cdots & g_{1,2k-1} & 0\\
  & g_{2,2} & \cdots & g_{2,k-1} & 0 & & g_{2,k+2}  & \cdots & g_{2,2k-1} & 0 \\
  & & \ddots & & \vdots & & & \ddots & \vdots \\
  & & & g_{k-1,k-1} & 0 & & & & g_{k-1,2k-1} & 0 \\
  & & & & 1 & & & & & 1
\end{pmatrix}
$$
by adding the $k$-th row times $g_{i,k}$ to the $i$-th row for
$i \in \left\lbrace 1,\hdots ,k-1\right\rbrace $.
Let $\tilde{g_i}$ be the $i$-th row of $\tilde{G}(z)$. Similarly to
before, from $\tilde{g_i}\tilde{g}_{k-1}^\top$ for
$i \in \left\lbrace 1,...,k-1\right\rbrace $, we get
$$
\begin{array}{ccc}
  g_{k-1,k-1} & = & g_{k-1,{2k-1}} \\
  g_{k-1,k-1}(g_{i,k-1}+g_{i,2k-1}) & = & 0 \\
\end{array} .
$$
Hence, using the same argument as before we get that $\tilde{G}(z)$ is
equivalent to
$$
\begin{pmatrix}
  g_{1,1} & \cdots & g_{1,k-2} & 0 & 0 & g_{1,k+1} & \cdots & g_{1,2k-2} & 0 & 0 \\
  & \ddots & \vdots & \vdots & \vdots & & \ddots & \vdots & \vdots & \vdots \\
  & & g_{k-2,k-2} & 0 & 0 & & & g_{k-2,2k-2} & 0 & 0 \\
  & & & 1 & 0 & & & & 1 & 0 \\
  & & & & 1 & & & & & 1
\end{pmatrix}.
$$
Repeating this reasoning,
we find that every double upper triangle generator matrix of a binary
self-dual code is equivalent to
$$
G^\prime(z) =
\begin{pmatrix}
  1 & & & & 1 & & &  \\
  & 1 & & & & 1  & & \\
  & & \ddots & & & & \ddots & \\
  & & & 1 & & & & 1
\end{pmatrix}.
$$
Conclusively, there exists exactly one binary self-dual $(2k,k)$
convolutional code that has a double triangular generator matrix,
namely $\mathcal{C}=rowspan(G^\prime(z))$.

\section{Construction of Self-Dual Convolutional Codes}
In this section,
we will generalize known construction methods for block codes to
construction methods for convolutional codes and will present examples
for binary self-dual convolutional codes of different lengths,
dimensions and degrees.
We will consider constructions, where already known self-dual
convolutional codes are used to find new ones. The first three results
are true over any polynomial ring over a general finite field and the
rest only over the polynomial ring over the binary field.
\begin{proposition}
  If $G(z) \in \mathbb F_q[z]^{k \times n}$ and
  $\tilde{G}(z) \in \mathbb F_q[z]^{k^\prime \times n^\prime}$ are
  generator matrices of two self-dual convolutional codes, then the
  code $\mathcal{C}^\prime$ generated by
$$G^\prime(z)=\begin{pmatrix}
  G(z) & 0 \\
  0 & \tilde{G}(z)
\end{pmatrix} \in \mathbb F_q[z]^{(k+k^\prime) \times (n+n^\prime)}$$
is also self-dual.
\end{proposition}
\begin{proof}
  We first compute the Smith normal form of $G^\prime(z)$:
$$
G^\prime(z)=
\begin{pmatrix}
  G(z) & 0 \\
  0 & \tilde{G}(z)
\end{pmatrix}
\sim
\begin{pmatrix}
  [I_k \ 0] & 0 \\
  0 & [I_{k^\prime} \ 0]
\end{pmatrix}
\sim
\begin{pmatrix}
  I_{k+k^\prime} & 0
\end{pmatrix},
$$
where the first equivalence follows from $G(z)$ and $\tilde{G}(z)$
generating non-catastrophic codes. It follows that also
$\mathcal{C}^\prime$ is non-catastrophic.  Finally
$$G^\prime(z)G^\prime(z)^\top=\begin{pmatrix}
  G(z) & 0 \\
  0 & \tilde{G}(z)
\end{pmatrix}
\begin{pmatrix}
  G(z) & 0 \\
  0 & \tilde{G}(z)
\end{pmatrix}^\top =
\begin{pmatrix}
  G(z)G(z)^\top & 0 \\
  0 & \tilde{G(}z)\tilde{G}(z)^\top
\end{pmatrix}
= 0$$ and the result follows with Theorem \ref{sdcrit}.
\end{proof}
The following proposition is a generalisation of a result presented in
\cite{15}.
\begin{proposition}
  Let $\mathcal{C}\subset\mathbb F_q[z]^n$ be a self-dual
  convolutional code with generator matrix $G(z)$. Let $r\in\mathbb N$
  and for $i=1,\hdots,n$, let $M_i(z) \in \mathbb F_q[z]^{n\times n}$
  be such that
$$M_i(z) M_i(z)^\top = \lambda_i I_n$$
for $\lambda_i$ invertible in $\mathbb F_q[z]$.  Let $A_1,\hdots,A_r$
be $n\times n$ permutations matrices. Then
$$G_r(z)=G(z)M_1(z)A_1\cdots M_r(z)A_r$$
generates a self-dual code.
\end{proposition}
\begin{proof}
  The idea of the proof is similar to that of \cite[Prop. 3.1]{15}.  \\
  We know that $G(z)$ has Smith form $[I_k \ 0]$ and $M_i(z)$ is
  unimodular as
$$M_i(z)(\lambda_i^{-1}M_i^\top(z))=\lambda_i^{-1}(M_i(z)M_i^\top(z))=\lambda_i^{-1}\lambda_i I_n=I_n.$$
Therefore also $G(z)M_1(z)$ has Smith form $[I_k \ 0]$.  By
definition, the $A_i$`s are also unimodular, implying that
$$G_r(z)=G(z)M_1(z)A_1\cdots M_r(z)A_r$$
has Smith form $[I_k \ 0]$ as well, which lets us conclude with
Theorem \ref{catcrit} that $G_r(z)$ generates a non-catastrophic code.
Finally due to $A_iA_i^\top=I_n$, we find
\begin{align*}
  G_r(z)G_r(z)^\top & = G(z)M_1(z)A_1 \cdots M_r(z) (A_rA_r^\top)M_r(z)^\top \cdots A_1^\top M_1(z)^\top G(z)^\top \\
                    & = \lambda_1 \lambda_2 \cdots \lambda_r G(z)G(z)^\top  =  0.
\end{align*}
Hence, with Theorem \ref{sdcrit} we can conclude that $G_r(z)$
generates a self-dual convolutional code.
\end{proof}
In Section 2.2, two construction methods for binary self-dual block codes were introduced in the building-up and Harada-Munemasa construction. We will discuss in the following pages, how they are best implemented for convolutional codes and which interesting properties their generalisations possess.\\

At first, we will consider the building-up construction, which we not
only generalize from block to convolutional codes but also from binary
fields to any finite field that has the property that $-1$ is a square
in this field.

\begin{proposition}[\textbf{Generalized Building-up construction}]
  Let $G(z)$ be a generator matrix of a self-dual $(2k,k)$
  convolutional code $\mathcal{C}$ over $\mathbb{F}_q[z]$. Let
  $q=p^l$, where $p=3 \ mod \ 4$ and $l$ being odd does not hold
  simultaneously. Further, let $a^2+b^2=0 \ mod \ p$ for some
  $a,b \in \mathbb{F}_q^*$, $f(z) \in \mathbb{F}_q[z]^{2k}$ with
  $f(z)f(z)^T=-(a^{-1})^2$ and define $y_i(z):=f(z)g_i(z)^T$ for
  $1\leq i \leq k$. Then
$$\tilde{G}(z)=
\begin{pmatrix}
  -a^{-1} & 0 & f(z) \\
  ay_1(z) & by_1(z) & \\
  \vdots & \vdots & G(z) \\
  ay_k(z) & by_k(z) &
\end{pmatrix} \in \mathbb{F}_q[z]^{(k+1) \times (2k+2)}
$$
generates a self-dual $(2k+2,k+1)$ convolutional code
$\tilde{\mathcal{C}}$.
\end{proposition}
\begin{proof}
  According to Theorem \ref{sdcrit}, we have to show that
  $\tilde{G}(z)$ generates a non-catastrophic code and that
  $\tilde{G}(z)\tilde{G}(z)^\top=0$.  Computing the Smith form of
  $\tilde{G}(z)$ and using that $G(z)$ has column Hermite form
  $[I_k \ 0]$ yields
  \begin{align*}
    \tilde{G}(z)&=
                  \begin{pmatrix}
                    -a^{-1} & 0 & f(z) \\
                    ay_1(z) & by_1(z) & \\
                    \vdots & \vdots & G(z) \\
                    ay_k(z) & by_k(z) &
                  \end{pmatrix}
                                        \sim
                                        \begin{pmatrix}
                                          -a^{-1} & 0 & f(z) \\
                                          0 & by_1(z) & \\
                                          \vdots & \vdots & G(z) \\
                                          0 & by_k(z) &
                                        \end{pmatrix}\\
                &\sim
                  \begin{pmatrix}
                    1 & 0 & 0 \\
                    0 & by_1(z) & \\
                    \vdots & \vdots & G(z) \\
                    0 & by_k(z) &
                  \end{pmatrix}
                                  \sim
                                  [I_{k+1} \ 0].
  \end{align*}
  Applying Theorem \ref{catcrit} we obtain that $\tilde{G}(z)$
  generates a non-catastrophic code.  Let $\tilde{g}_i(z)$ be the
  $i$-th row of $\tilde{G}(z)$ and $g_i(z)$ be the $i$-th row of
  $G(z)$. To see that $\tilde{\mathcal{C}}$ is self-dual, we compute
  \begin{align*}
    \tilde{g}_1(z)\tilde{g}_1(z)^\top&=(-a^{-1})^2+f(z)f(z)^\top=(a^{-1})^2-(a^{-1})^2=0\\
    \tilde{g}_1(z)\tilde{g}_i(z)^\top&=-a^{-1}ay_i(z)+f(z)g_i(z)=-y_i(z)+y_i(z)=0\\
    \tilde{g}_{i+1}(z)\tilde{g}_{j+1}(z)^\top&=a^2y_i(z)y_j(z)+b^2y_i(z)y_j(z)+g_i(z)g_j(z)^\top=g_i(z)g_j(z)^\top=0
  \end{align*}
\end{proof}

\begin{remark}
  If we set $p=2$ in the preceding theorem, one has $a=b=1$ and
  obtains a straightforward generalization of Theorem \ref{bu}.
\end{remark}

\begin{example}
  From Section 5.2, we know that
$$G_1(z)=
\begin{pmatrix}
  1 & 1 & 1 & 1 \\
  0 & 1 & z+1 & z
\end{pmatrix}$$
generates a binary self-dual convolutional code. We apply the generalized building-up construction with $f(z)=(1,z,z^2,z^2+z) \in \mathbb{F}_2[z]^4$.\\
The condition of $f(z)f(z)^\top=a=1$ is fulfilled and
\begin{align*}
  y_1(z) = f(z)(1,1,1,1)^\top = 1, \qquad y_2(z)  = f(z)(0,1,z+1,z)^\top  = z.
\end{align*}
Therefore
$$G_2(z)=
\begin{pmatrix}
  1 & 0 & \horzbar & f(z) & \horzbar & \horzbar \\
  y_1(z) & y_1(z) & 1 & 1 & 1 & 1 \\
  y_2(z) & y_2(z) & 0 & 1 & z+1 & z
\end{pmatrix}
=
\begin{pmatrix}
  1 & 0 & 1 & z & z^2 & z^2+z \\
  1 & 1 & 1 & 1 & 1 & 1 \\
  z & z & 0 & 1 & z+1 & z
\end{pmatrix}$$ generates a binary self-dual convolutional code.
\end{example}
Next we show that Theorem \ref{abu} cannot be generalized to
convolutional codes.
\begin{lemma}\label{nbu}
  Not every binary self-dual convolutional code with free distance
  $d_{free} > 2$ can be constructed with the generalized building-up
  construction (up to permutation of columns).
\end{lemma}
\begin{proof}
  By Section 5.2,
$$
G(z)=
\begin{pmatrix}
  0 & z^2+z+1 & z & z^2+1 \\
  1 & 1 & 1 & 1
\end{pmatrix} \in \mathbb{F}_2[z]^{2 \times 4}
$$
generates a binary self-dual convolutional code. \\
We claim that $G(z)$ cannot be constructed with the generalized building-up construction, because every codeword of the corresponding code $\mathcal{C}$ has the property that if one of its entries is equal to $0$, then none of its entries can be equal to $1$.\\
To see this, assume that for some $a(z),b(z)\in \mathbb{F}_2[z]$,
 $$c(z)=a(z)(1,1,1,1)+b(z)(0,z^2+z+1,z,z^2+1)$$ has one entry that is equal to $0$. This implies $a(z)=b(z)d(z)$ for some $d(z)\in\{0,z^2+z+1,z,z^2+1\}$ and hence, $\gcd(a(z),b(z))=a(z)\neq 1$. Consequently, $a(z)+b(z)d(z)\neq 1$ (by the Lemma of Bezout).\\
%
 We conclude that there is no codeword that has a entry equal to $0$
 as well as a entry equal to $1$, meaning that the top-left $(1 \; 0)$
 in the building-up construction can never be achieved by row
 operations or column permutations on $G(z)$. Hence, $\mathcal{C}$
 cannot be constructed with the generalized building-up construction.
 \\
 It remains to show that $d_{free}(\mathcal{C})>2$. As a consequence
 of Lemma \ref{ew}, every binary self-dual code has an even free
 distance and as the free distance is non-zero, we just need to show
 that there is no codeword of weight $2$.
 \\
 Assume conversely that there exists
 $c(z)=(c_1(z),c_2(z),c_3(z),c_4(z))\in \mathcal{C}$ with $wt(c(z))=2$
 and write
 \begin{align*}
   c(z&)=a(z)(1,1,1,1)+b(z)(0,z^2+z+1,z,z^2+1)\\
      &=(a(z),\ a(z) + (z^2+z+1)b(z),\ a(z) + zb(z),\  a(z) +  (z^2+1)b(z)). 
 \end{align*}
 If $a(z)=0$ and $b(z)=0$, then $wt(c(z))=0$ and if $a(z)=0$ and
 $b(z)\neq 0$, then $wt(c(z))\geq 3$, i.e. both of these cases lead to
 a contradiction.

 Hence $a(z) \neq 0$. If $b(z)=0$, $wt(c(z))=4\cdot wt(a(z))\geq 4$,
 i.e. $b(z)\neq 0$.  But this implies that only one of the entries
 $c_2(z), c_3(z), c_4(z)$ can be zero, i.e., $wt(c(z))\geq 3$, again a
 contradiction.
%
%
 We conclude that $\mathcal{C}$ has free distance $d_{free}>2$ and
 cannot be constructed with the generalized building-up construction.
\end{proof}
As a consequence, the full classification of (binary) self-dual
convolutional codes cannot be done by using the generalized
building-up construction. In the following, we will investigate
generalizations of the Harada-Munemasa construction.

\begin{proposition} \label{hmo} Let
  $G(z)\in \mathbb{F}_2[z]^{k\times 2k}$ be a generator matrix of a
  binary self-dual $(2k,k)$ convolutional code and
  $a_i(z)\in \mathbb{F}_2[z]$, for
  $i\in \left\lbrace 1,\hdots,k\right\rbrace $. Then
$$
\tilde{G}(z)=
\begin{pmatrix}
  a_1(z) & a_1(z) & \\
  \vdots & \vdots & G(z) \\
  a_k(z) & a_k(z) &
\end{pmatrix} \in \mathbb{F}_2[z]^{k \times (2k+2)}
$$
generates a binary self-orthogonal convolutional code
$\tilde{\mathcal{C}}$.
\end{proposition}
\begin{proof}
  Let $g_i(z),g_j(z)$ be any two rows of $G(z)$ and $\tilde{g_i}(z)$
  and $\tilde{g_j}(z)$ be the two corresponding rows of
  $\tilde{G}(z)$. Then
$$\tilde{g_i}(z)\tilde{g_j}(z)^\top=a_i(z)a_j(z)+a_i(z)a_j(z)+g_i(z)g_j(z)^\top=g_i(z)g_j(z)^\top=0.$$
\end{proof}
\begin{defn}\label{com}
  Let $G(z)\in \mathbb{F}_2[z]^{k \times 2k}$ be a generator matrix of
  a binary self-dual $(2k,k)$ convolutional code and for some
  $a_i(z)\in \mathbb{F}_2[z]$
$$\tilde{G}(z)=
\begin{pmatrix}
  a_1(z) & a_1(z) & \\
  \vdots & \vdots & G(z) \\
  a_k(z) & a_k(z) &
\end{pmatrix} \in \mathbb{F}_2[z]^{k \times (2k+2)}
$$
like in Proposition \ref{hmo}. If there exists
$f(z)\in \mathbb{F}_2[z]^{2k+2}$ such that
$$
G_1(z)=
\begin{pmatrix}
  \horzbar & f(z) & \horzbar \\
  a_1(z) & a_1(z) & \\
  \vdots & \vdots & G(z) \\
  a_k(z) & a_k(z) & \\
\end{pmatrix} \in \mathbb{F}_2[z]^{(k+1) \times (2k+2)}
$$
generates a self-dual code, then we say that $G_1(z)$ is a
\textbf{self-dual completion} or just a \textbf{completion} of
$\tilde{G}(z)$.
\end{defn}
Naturally, we want to figure out when and how a self-dual completion
can be found.
\begin{lemma}\label{10}
  Choosing $f(z)=(1,1,0,\hdots,0)$ in Definition \ref{com} leads to a
  generator matrix of a self-dual code. More concretely,
$$
G_1(z)=
\begin{pmatrix}
  1 & 1 & 0 \\
  a_1(z) & a_1(z) & \\
  \vdots & \vdots & G(z) \\
  a_k(z) & a_k(z) & \\
\end{pmatrix} \in \mathbb{F}_2[z]^{(k+1) \times (2k+2)}
$$
is a self-dual completion for any $a_i(z)\in \mathbb{F}_2[z]$.
\end{lemma}
\begin{proof}
  Obviously, $G_1(z)$ and
$$
G_2(z)=\begin{pmatrix}
  1 & 1 & 0\\
  0 & 0 & \\
  \vdots & \vdots & G(z) \\
  0 & 0 &
\end{pmatrix}
$$
generate the same code, making the choice of the $a_i(z)$
pointless. But still
$$G_2(z)G_2(z)^\top=0$$
and
$$G_2(z) \sim [I_{k+1} \ 0],$$
implying that the code generated by $G_2(z)$ is self-dual.
\end{proof}
Because of the preceding lemma we make the following definition.
\begin{defn}
  Let
$$
G_1(z)=
\begin{pmatrix}
  \horzbar & f(z) & \horzbar \\
  a_1(z) & a_1(z) & \\
  \vdots & \vdots & G(z) \\
  a_k(z) & a_k(z) &
\end{pmatrix} \in \mathbb{F}_2[z]^{(k+1) \times (2k+2)}
$$
be a self-dual completion. Then we say that the completion was
\textbf{trivial} if $G_1(z)$ and
$$
\begin{pmatrix}
  1 & 1 & 0 \\
  a_1(z) & a_1(z) & \\
  \vdots & \vdots & G(z) \\
  a_k(z) & a_k(z) &
\end{pmatrix} \in \mathbb{F}_2[z]^{(k+1) \times (2k+2)}
$$
generate the same code. If they do not, we call the completion
\textbf{non-trivial}.
\end{defn}

\begin{remark}\label{r10}
  An important observation here is that $(1,1,0,\hdots,0)$ can never
  be a linear combination of the last $k$ rows of $G_1(z)$, because
  the rows of $G(z)$ are linearly independent.
\end{remark}

In the following example, we demonstrate how the choice of the
$a_i(z)$ affects the existence of non-trivial self-dual completions.
\begin{example}
  We know that $(1 \; 1)$ is a binary self-dual $(2,1)$ convolutional
  code (actually the only one). We implement the Harada-Munemasa
  construction for $a_1(z)=z$ and show that for this choice of
  $a_1(z)$, there are only trivial self-dual completions.  By
  Proposition \ref{hmo}, $
  \begin{pmatrix}
    z & z & 1 & 1
  \end{pmatrix}
  $ generates a self-orthogonal code. We are now looking for
  $f(z)=(f_1,f_2,f_3,f_4)\in \mathbb{F}_2[z]^4$ such that
$$
G(z)=
\begin{pmatrix}
  f_1 & f_2 & f_3 & f_4 \\
  z & z & 1 & 1
\end{pmatrix}
$$
generates a self-dual code $\mathcal{C}$.  By Lemma \ref{11}, the all
one vector must be part of the code. Therefore, there is
$b(z)=(b_1,b_2) \in \mathbb{F}_2[z]^2$ such that
$$b(z)G(z)=(1,1,1,1).$$
Notice that necessarily $b_1 \neq 0$.  This translates into the
following system of equations over $\mathbb{F}_2[z]$
$$
\begin{array}{ccc}
  b_1f_1+zb_2 & = & 1 \\
  b_1f_2+zb_2 & = & 1 \\
  b_1f_3 + b_2 & = & 1 \\
  b_1f_4 + b_2 & = & 1
\end{array}.
$$
Since $b_1 \neq 0$, by adding the last two equations we find that
$f_4=f_3$ and the last equation is equivalent to
$$b_2=b_1f_4+1.$$
Hence the system simplifies into
$$
\begin{array}{ccc}
  b_1f_1+z(b_1f_4+1) & = & 1 \\
  b_1f_2+z(b_1f_4+1) & = & 1 
\end{array}
$$
and similarly we find $f_1=f_2$ by adding the two equations.
\\
Finally $b_1f_1+z(b_1f_4+1)=1$ is equivalent to
$$b_1(f_1+f_4z)=1+z$$
and we conclude that either $b_1=1$ or $b_1=z+1$.
\\
Let $b_1=1$, then $f_1=1+z+zf_4$ and the generator matrix is given by
$$
G(z)=
\begin{pmatrix}
  1+z+f_4z & 1+z+f_4z & f_4 & f_4 \\
  z & z & 1 & 1
\end{pmatrix}
$$
depending only on $f_4 \in \mathbb{F}_2[z]$.
Computing its Smith form, one obtains
\begin{align*}
  \begin{pmatrix}
    1+z+f_4z & 1+z+f_4z & f_4 & f_4 \\
    z & z & 1 & 1
  \end{pmatrix}
      &\sim 
        \begin{pmatrix}
          1+z & 1+z & 0 & 0 \\
          z & z & 1 & 1
        \end{pmatrix}\\
             &\sim
               \begin{pmatrix}
                 1+z & 0 & 0 & 0 \\
                 0 & 1 & 0 & 0
               \end{pmatrix}
\end{align*}
Since the Smith form of $G(z)$ is not $[I_k \ 0]$, $G(z)$ generates a
catastrophic and therefore not self-dual code.
\\
Let $b_1=z+1$, then $f_1=1+f_4z$ and
$$
G(z)=
\begin{pmatrix}
  1+f_4z & 1+f_4z & f_4 & f_4 \\
  z & z & 1 & 1
\end{pmatrix}.
$$
But $G(z)$ and
$$
\begin{pmatrix}
  1 & 1 & 0 & 0 \\
  z & z & 1 & 1
\end{pmatrix}
$$ 
generate the same code, implying that $G(z)$ was a trivial self-dual
completion.
\\
We conclude that for this choice of $a_1(z)$ there are only trivial
self-dual completions.
\\
Now let $a_1(z)=1$. Then $
\begin{pmatrix}
  1 & 1 & 1 & 1
\end{pmatrix}
$ generates a self-orthogonal code and a non-trivial self-dual
completion exists in the form of
$$
\begin{pmatrix}
  0 & z^2+z+1 & z & z^2+1 \\
  1 & 1 & 1 & 1
\end{pmatrix}.
$$
This generates a self-dual code by Section 5.2.
\end{example}
Therefore, the existence of non-trivial completions depends on the
choice of $a_i(z)$. This naturally steers to the search for viable
conditions on the $a_i(z)$ that admit non-trivial self-dual
completions.
\begin{theorem}\label{main}
  Let
$$
\tilde{G}(z)=
\begin{pmatrix}
  a_1(z) & a_1(z) & \\
  \vdots & \vdots & G(z) \\
  a_k(z) & a_k(z) &
\end{pmatrix} \in \mathbb{F}_2[z]^{k \times (2k+2)}
$$
Then a non-trivial self-dual completion exists if and only if
$$(1,\hdots,1) \in rowspan(\tilde{G}(z)).$$
\end{theorem}
\begin{proof}
  Let $\tilde{\mathcal{C}}=rowspan(\tilde{G}(z))$,
$$
G_1(z)=
\begin{pmatrix}
  \horzbar & f(z) & \horzbar \\
  a_1(z) & a_1(z) & \\
  \vdots & \vdots & G(z) \\
  a_k(z) & a_k(z) &
\end{pmatrix}
$$
be a self-dual completion for some $f(z)\in\mathbb{F}_2[z]^{2k+2}$ and
$\mathcal{C}_1=rowspan(G_1(z))$.

Obviously, all self-dual completions of $\tilde{\mathcal{C}}$ can be
found by looking at all possible new rows
$f(z) \in \bigslant{\tilde{\mathcal{C}}^\perp}{\tilde{\mathcal{C}}}$
and checking the resulting $\mathcal{C}_1$ for self-duality.

The main fact we will use to prove this theorem is that if
$(1,\hdots,1)\in \tilde{\mathcal{C}},$ then every possible new row
$f(z) \in \tilde{\mathcal{C}}^\perp$ that we want to add to
$\tilde{G}(z)$, fulfils
$$f(z)f(z)^\top=(f(z)(1,\hdots,1)^\top)^2=0,$$
as $(1,\hdots,1)\in \tilde{\mathcal{C}}$, which implies that
each pair of rows of $G_1(z)$ is orthogonal.
\\
$(\Longrightarrow)$ We will prove that
$(1,\hdots,1)\not \in \tilde{\mathcal{C}}$ implies that there exists
exactly one self-dual completion (up to equivalent generator matrices)
and it is the trivial one.
\\
If $\mathcal{C}_1$ with generator matrix $G_1(z)$ is assumed to be
self-dual (and hence, also non-catastrophic), then
$(1,\hdots,1)\in\mathcal{C}_1$. Thus,
$\begin{pmatrix} 1 & \cdots & 1\\ & \tilde{G}(z) &
\end{pmatrix}$ is a generator matrix for $\mathcal{C}_1$. As
$\mathcal{C}$ is self-dual, $(1,\hdots,1)\in rowspan(G(z))$ and
therefore, $\mathcal{C}_1$ has also a generator matrix of the form
$$\begin{pmatrix} b(z) & b(z) & 0 & \cdots & 0\\ a_1(z) & a_1(z) & \\
  \vdots & \vdots & &G(z)\\ a_k(z) & a_k(z) &
\end{pmatrix},$$
for some $b(z)\in \mathbb{F}_2[z]$. As $\mathcal{C}_1$ is non-catastrophic $b(z)=1$ and hence, the self-dual completion is trivial.\\
$(\Longleftarrow)$ We have
$dim(\bigslant{\tilde{\mathcal{C}}^\perp}{\tilde{\mathcal{C}}})=n-k-k=2$. Let
$$\bigslant{\tilde{\mathcal{C}}^\perp}{\tilde{\mathcal{C}}}=<(1,1,0,\hdots),f(z)>$$
for some $f(z)=(f_1(z),\hdots,f_{2k+2}) \in
\mathbb{F}_2[z]^{2k+2}$. As $\tilde{\mathcal{C}}^{\perp}$ is
non-catstrophic and $f(z)$ can be completed to a basis of this code,
$\gcd(f_1(z),\hdots,f_{2k+2}(z))=1$. We want to show that the code
$\mathcal{C}_1$ generated by
$$G_1(z)=
\begin{pmatrix}
  f(z) \\
  \tilde{G}(z)
\end{pmatrix}
$$
is a non-trivial self-dual completion.
\\
First, we observe that $G_1(z)G_1(z)^{\top}=0$ as $\tilde{\mathcal{C}}$ is self-orthogonal, $f(z)\in\tilde{\mathcal{C}}^{\perp}$ and $(1,\hdots,1)\in\tilde{\mathcal{C}}$ implies $f(z)f(z)^\top=0$.\\
Next, we will show that $\mathcal{C}_1$ is non-catastrophic, which
then shows that $\mathcal{C}_1$ is self-dual. Applying column
operations, we obtain
$$G_1(z)\sim\begin{pmatrix}
  \tilde{f}_1(z) & \cdots & \tilde{f}_{k+2}(z) & \tilde{f}_{k+3}(z) &
  \cdots & \tilde{f}_{2k+2}(z) \\ & 0 & & & I_k\end{pmatrix},$$ where
$\tilde{f}_1(z), \hdots, \tilde{f}_{2k+2}(z)\in\mathbb F_2[z]$ with
$$\gcd(\tilde{f}_1(z), \hdots,
\tilde{f}_{2k+2}(z))=\gcd(f_1(z),\hdots,f_{2k+2}(z))=1.$$
Consequently, $$G_1(z)\sim\begin{pmatrix}
  1 & 0  & 0\\
  0 & I_k & 0\end{pmatrix}$$ and $\mathcal{C}_1$ is non-catastrophic.
\\
Finally, we will show that
$(1,1,0,\hdots,0) \not \in \mathcal{C}_1^\perp$.
\\
Assume conversely that $(1,1,0,\hdots,0)\in \mathcal{C}_1^\perp$. As
$(1,\hdots,1)\in\tilde{\mathcal{C}}$ implies
$f(z)\in\mathcal{C}_1^{\perp}$ and one has
$\tilde{\mathcal{C}}\subset\mathcal{C}_1=\mathcal{C}_1^{\perp}$, one
gets that
$\mathcal{C}_1^{\perp}=\tilde{\mathcal{C}}^{\perp}$. However,
$dim(\mathcal{C}_1^{\perp})=k+1\neq
k+2=dim(\tilde{\mathcal{C}}^{\perp})$, a contradiction.

Conclusively, we showed $(1,1,0,\hdots,0)\not \in
\mathcal{C}_1^\perp$ and
$\mathcal{C}_1=\mathcal{C}_1^\perp$, i.e. we have a non-trivial
completion to a self-dual convolutional code.
\end{proof}
Thus, we can find non-trivial completions with the following two
steps:
\begin{itemize}
\item[1.] Make sure that $(1,\hdots,1)\in
  rowspan(\tilde{G}(z))=\tilde{\mathcal{C}}$, which means we must
  establish that the system
  \begin{equation}\label{1is}
    b(z) \tilde{G}(z)=(1,\hdots,1)\in \mathbb{F}_2[z]^{2k+2},
  \end{equation}
  where $b(z)=(b_1,\hdots,b_k)\in
  \mathbb{F}_2[z]^k$, has a solution for $b(z)$.
\item[2.] Find a vector $f(z) \in
  \tilde{\mathcal{C}}^\perp$ that forms a basis of
  $\bigslant{\tilde{\mathcal{C}}^\perp}{\tilde{\mathcal{C}}}$ with
  $(1,1,0,\hdots,0)$ and then
$$G_1(z)=
\begin{pmatrix}
  \horzbar & f(z) & \horzbar \\
  a_1(z) & a_1(z) & \\
  \vdots & \vdots & G(z) \\
  a_k(z) & a_k(z) &
\end{pmatrix}$$ generates a binary self-dual convolutional code.
\end{itemize}
We will further investigate these two steps in the following.
\begin{remark}\label{rcp}
  A way to make sure that \eqref{1is} has a solution is by finding a
  solution to
$$b(z)G(z)=(1,\hdots,1) \in \mathbb{F}_2[z]^{2k},$$
which always exists since
$G(z)$ generates a binary self-dual code (see Lemma \ref{11}), and
then choosing the $a_i(z)\in \mathbb{F}_2[z]$ such that
$$\sum_{i=1}^k a_i(z)b_i(z)=1,$$ 
which is possible since
$b(z)G(z)=1$ implies
$gcd(b_1(z),\hdots,b_k(z))=1$. In particular,
$a_i(z)=g_i(z)$ for $i=1,\hdots,k$ yields such a solution.
\end{remark}
The following lemma shows, how to do the second step, i.e. how to find
a suitable $f(z)$.
\begin{lemma}
  Let $f(z)=(f_1(z),\hdots,f_{2k+2}(z)) \in
  \bigslant{\tilde{\mathcal{C}}^\perp}{<\tilde{\mathcal{C}},(1,1,0,\hdots,0)>}$
  and
$$gcd(f_1(z),\hdots,f_{2k+2}(z))=1,$$
then
$$<f(z),(1,1,0,\hdots,0)>=\bigslant{\tilde{\mathcal{C}}^\perp}{\tilde{\mathcal{C}}}.$$
\end{lemma}
\begin{proof}
  Assume conversely that there exists $h(z)\in
  \bigslant{\tilde{\mathcal{C}}^\perp}{\tilde{\mathcal{C}}}$ such that
  $h(z)\neq f(z)$ and
$$<h(z),(1,1,0,\hdots,0)>=\bigslant{\tilde{\mathcal{C}}^\perp}{\tilde{\mathcal{C}}}.$$
As $f(z) \in
\bigslant{\tilde{\mathcal{C}}^\perp}{<\tilde{\mathcal{C}},(1,1,0,\hdots,0)>}=<h(z)>$,
$f(z)$ must be a multiple of $h(z)$. But since the entries of
$f(z)$ have no common non-trivial divisor,
$f(z)=h(z)$ follows immediately.
\end{proof}
\begin{example}
  Let
$$
G(z)=
\begin{pmatrix}
  0 & z^2+z+1 & z & z^2+1 \\
  1 & 1 & 1 & 1
\end{pmatrix} \in \mathbb{F}_2[z]^{2 \times 4}.
$$
We know by Section 5.2, that
$G(z)$ generates a self-dual convolutional code.
\\
By inspection,
$$b(z)G(z)=(1,1,1,1)$$
holds exactly for $b(z)=(0,1)$. Hence we need
$$0\cdot a_1(z)+1\cdot a_2(z)=1$$
or equivalently $a_2(z)=1$ and
$a_1(z)$ may be chosen arbitrarily. So let
$a_1(z)=z^2+1$ and $a_2(z)=1$, then
$$
\tilde{G}(z) =
\begin{pmatrix}
  z^2+1 & z^2+1 & 0 & z^2+z+1 & z & z^2+1 \\
  1 & 1 & 1 & 1 & 1 & 1
\end{pmatrix}.
$$
We are now looking for $f(z)\in
\bigslant{\tilde{\mathcal{C}}^\perp}{<\tilde{\mathcal{C}},(1,1,0,\hdots,0)>}$. If
$f_1(z) \neq f_2(z)$, this guarantees $f(z) \not \in
<\tilde{\mathcal{C}},(1,1,0,0,0,0)>$.
\\
Further, we need
$f(z)$ to be orthogonal to both rows of $\tilde{G}(z)$, i.e.,
$$
\begin{array}{ccc}
  f(z)(z^2+1,z^2+1,0,z^2+z+1,z,z^2+1)^\top & = & 0 \\
  f(z)(1,1,1,1,1,1)^\top & =& 0
\end{array}
$$
and see that this holds for $f(z)=(0,1,0,0,0,1).$
One obtains that
$$
G_1(z) =
\begin{pmatrix}
  0 & 1 & 0 & 0 & 0 & 1 \\
  z^2+1 & z^2+1 & 0 & z^2+z+1 & z & z^2+1 \\
  1 & 1 & 1 & 1 & 1 & 1
\end{pmatrix}
$$
generates a binary self-dual convolutional code.
\end{example}
An immediate consequence of Remark \ref{rcp} is that there must be a
linear combination of the $a_i(z)$ that is equal to
$1$. But this is not possible, if they have a common non-trivial
divisor, giving us the following corollary.
\begin{corollary}
  Let
$$
G_1(z)=
\begin{pmatrix}
  \horzbar & f(z) & \horzbar \\
  a_1(z) & a_1(z) & \\
  \vdots & \vdots & G(z) \\
  a_k(z) & a_k(z) &
\end{pmatrix} \in \mathbb{F}_2[z]^{(k+1) \times (2k+2)}
$$
be a self-dual completion, where $gcd(a_1(z),\hdots,a_k(z))\neq
1.$ Then, the completion was trivial.
\end{corollary}
The converse is not true in general as can be seen in the next
example.  Furthermore, we illustrate in the upcoming example that a
good choice of the
$a_i(z)$ (``good choice" meaning that a non-trivial completion exists)
differs for equivalent generator matrices.
\begin{example}
  By Theorem \ref{main}, if we chose $a_1=a_2=1$ for
$$\begin{pmatrix}
  1 & 1 & 0 & 0 \\
  0 & 0 & 1 & 1 \\
\end{pmatrix} \in \mathbb{F}_2[z]^{2 \times 4},
$$
to obtain
$$
\tilde{G}(z)=\begin{pmatrix}
  1 & 1 & 1 & 1 & 0 & 0 \\
  1 & 1 & 0 & 0 & 1 & 1 \\
\end{pmatrix},
$$
then $\tilde{G}(z)$ can only be trivially completed as $(1,1,1,1,1,1)
\not \in
rowspan(\tilde{G}(z))$. Now, we claim that another generator matrix of
the same code, namely
$$G(z)=
\begin{pmatrix}
  1 & 1 & 1 & 1 \\
  0 & 0 & 1 & 1 \\
\end{pmatrix} \in \mathbb{F}_2[z]^{2 \times 4}
$$
has a non-trivial completion for $a_1(z)=a_2(z)=1$ in
$$G_1(z)=
\begin{pmatrix}
  h(z) & h(z)+1 & 0 & 0 & 0 & 1 \\
  1 & 1 & 1 & 1 & 1 & 1 \\
  1 & 1 & 0 & 0 & 1 & 1
\end{pmatrix},
$$
where $h(z)\in \mathbb{F}_2[z]$.
$G_1(z)$ generates a non-catastrophic code, since
$$
det\left[
  \begin{pmatrix}
    h(z) & h(z)+1 & 0 \\
    1 & 1 & 1 \\
    1 & 1 & 0
  \end{pmatrix}\right]=h(z)+1+h(z)=1.
$$
Moreover, $G_1(z)G_1(z)^\top=0$ by inspection. Hence,
$G_1(z)$ generates a self-dual code by Theorem \ref{sdcrit}.
\\
So the only thing left to show is non-triviality.
\\
Assume by contradiction that there exists $b(z)=(b_1,b_2,b_3)\in
\mathbb{F}_2[z]^3$ such that
$$b(z)G_1(z)=(1,1,0,0,0,0),$$
or rather, in terms of a system of equations,
$$
\begin{array}{ccc}
  b_1h+b_2+b_3 & = & 1 \\
  b_1(h+1)+b_2+b_3 & = & 1 \\
  b_2 & = & 0\\
  b_2 & = & 0 \\
  b_2+b_3 & = & 0\\
  b_1+b_2+b_3 & = & 0 
\end{array}.
$$
But the last three equations imply
$b_1=b_2=b_3=0$, which is a contradiction and we conclude that the
completion was non-trivial.
\end{example}
Furthermore, if there are two generator matrices of the same code that
differ by given row operations, then from a good choice of the
$a_i(z)$ for one matrix, we can find a good choice for the other
matrix by applying the given row operations to the already known good
choice $(a_1(z),\hdots,a_k(z))^\top$.
\\
We were previously asking the question, whether all binary self-dual
convolutional codes with
$d_{free}>2$ can be constructed with the generalized Harada-Munemasa
construction. This question will be left unanswered for codes of
lengths greater than 4.
\begin{lemma}
  All binary self-dual
  $(4,2)$ convolutional codes can be constructed with the generalized
  Harada-Munemasa construction.
\end{lemma}
\begin{proof}
  Let $\mathcal{C}$ be a binary self-dual
  $(4,2)$ convolutional code. Then by Section 5.2, there exists a
  generator matrix of the form
$$
G_1(z) =
\begin{pmatrix}
  0 & g(z)+h(z) & g(z) & h(z) \\
  1 & 1 & 1 & 1
\end{pmatrix}
$$
for some $g(z),h(z)\in \mathbb{F}_2[z]$.
\\
\\
But now by setting $G(z)=(1 \; 1)$,
$a_1(z)=1$ and
$f(z)=(0,g(z)+h(z),g(z),h(z))$, we may construct
$G_1(z)$ and therefore
$\mathcal{C}$ with the generalized Harada-Munemasa construction.
\end{proof}
Finally, we want to connect the generalized building-up construction
and the generalized Harada-Munemasa construction.

\begin{lemma}
  Every code that was constructed with the generalized building-up
  construction, can be constructed with the generalized
  Harada-Munemasa construction, but not vice versa.
\end{lemma}
\begin{proof}
  Let
  $G(z)$ be a generator matrix of a binary self-dual convolutional
  code and let
$$
G_1(z) =
\begin{pmatrix}
  1 & 0 & f(z) \\
  y_1(z) & y_1(z) & \\
  \vdots & \vdots & G(z) \\
  y_k(z) & y_k(z) &
\end{pmatrix}
$$
be a generator matrix constructed with the building-up construction
for some $f(z)\in
\mathbb{F}_2[z]^{2k}$. To construct the same code with the
Harada-Munemasa construction, just set $a_i(z):=y_i(z)$ to get
$$\tilde{G}(z)
=\begin{pmatrix}
  a_1(z) & a_1(z) & \\
  \vdots & \vdots & G(z) \\
  a_k(z) & a_k(z) &
\end{pmatrix}=\begin{pmatrix}
  y_1(z) & y_1(z) & \\
  \vdots & \vdots & G(z) \\
  y_k(z) & y_k(z) &
\end{pmatrix}
$$
and then add $(1,0,f(z))$ as a new row.
\\
One the other hand, in Lemma \ref{nbu}, it was shown that
$$
\begin{pmatrix}
  0 & z^2+z+1 & z & z^2+1 \\
  1 & 1 & 1 & 1
\end{pmatrix}
$$
cannot be constructed using the generalized building-up
construction. But by the previous lemma, all binary self-dual
$(4,2)$ convolutional codes can be constructed with the generalized
Harada-Munemasa construction.
\end{proof}

\section{Conclusion}
We started with finding equivalent conditions to self-duality. This
was the foundation for all other obtained results. First, we used it
to fully classify all self-dual $(2,1)$ convolutional codes, all
binary self-dual $(4,2)$ convolutional codes, all self-dual
convolutional codes with double diagonal generator matrices and all
binary self-dual convolutional codes with double triangular generator
matrices.

Then, we investigated the construction of self-dual convolutional
codes, where the building-up construction and the Harada-Munemasa
construction were generalized. For the latter, conditions on the
$a_i(z)$ for the existence of non-trivial self-dual completions were
established and the generalized Harada-Munemasa construction was shown
to be able to construct strictly more binary codes.
\\
Further, we presented a binary self-dual convolutional code with free
distance $d_{free}>2$, which cannot be constructed with the
generalized building-up construction. As a consequence, the
generalized building-up construction is not viable for the full
classification of binary self-dual convolutional codes, like it is for
block codes. For the generalized Harada-Munemasa construction this
question about viability is still open for codes of lengths greater
than 4.

\section*{Acknowledgments}
The authors acknowledge the support of Swiss National Science
Foundation grant n. 188430.

\bibliography{mybibfile}

\end{document}